\pdfoutput=1

\documentclass{article}

\usepackage{arxiv}

\usepackage[utf8]{inputenc} % allow utf-8 input
\usepackage[T1]{fontenc}    % use 8-bit T1 fonts
\usepackage{url}            % simple URL typesetting
\usepackage{booktabs}       % professional-quality tables
\usepackage{amsfonts}       % blackboard math symbols
\usepackage{nicefrac}       % compact symbols for 1/2, etc.
\usepackage{microtype}      % microtypography

\RequirePackage[colorlinks]{hyperref}

\usepackage[autostyle=true]{csquotes} % Required to generate language-dependent quotes in the bibliography

\usepackage{framed}
\usepackage{bm}
\usepackage{graphics}
\usepackage{bbm}
\usepackage{enumitem}
\usepackage{amsmath}
\usepackage{amsfonts}
\usepackage{amssymb}
\usepackage{amsthm}
\usepackage{mathtools}
\usepackage{subcaption}
\usepackage{xcolor}
\usepackage{autonum}
\RequirePackage{scrextend}
\changefontsizes[14]{11}

\usepackage[natbibapa]{apacite}
\usepackage{natbib}
\def\bSig\mathbf{\Sigma}

\newcommand{\pr}{\mathbb{P}}
\newcommand{\E}{\mathbb{E}}

\newcommand{\R}{\mathbb{R}}

\newcommand{\one}{\mathbbm{1}}
\newcommand{\setA}{\mathcal{A}}

\newcommand{\setS}{\mathcal{S}}

\newcommand{\setN}{\mathcal{N}}
\newcommand{\setC}{\mathcal{C}}

\newcommand{\setX}{\mathcal{X}}

\newcommand{\thetahat}{\hat{\vartheta}}

\newcommand{\setP}{\mathcal{P}}

\newcommand{\setH}{\mathcal{H}}

\newcommand{\X}{\mathfrak{X}}
\newcommand{\Xmat}{\bm{\mathfrak{X}}}
\newcommand{\Cmat}{\bm{\mathfrak{C}}}
\newcommand{\tra}{^\top}
\newcommand{\given}{\, |\,}

\DeclareMathOperator{\var}{var}
\DeclareMathOperator{\cov}{cov}

\DeclareMathOperator{\diag}{diag}

\DeclareMathOperator{\argmin}{argmin}

\DeclareMathOperator{\CR}{\bold{CR}}
\DeclareMathOperator{\Bin}{Bin}
\DeclareMathOperator{\Dir}{Dir}
\DeclareMathOperator{\Beta}{Beta}
\DeclareMathOperator{\mBeta}{mBeta}

\DeclareMathOperator{\BCP}{BCP}

\newtheorem{definition}{Definition}
\newtheorem{example}{Example}
\newtheorem{lemma}{Lemma}
\newtheorem{theorem}{Theorem}
\newtheorem{prop}[theorem]{Proposition}

\mathtoolsset{
	showonlyrefs,
}

\title{Simultaneous Inference for Multiple Proportions: A Multivariate Beta-Binomial Model}
\author{
  Max Westphal\thanks{Correspondence to: Max Westphal, \href{mailto:mwestphal@uni-bremen.de}{mwestphal@uni-bremen.de},
  \url{https://orcid.org/0000-0002-8488-758X}}\\
  Institute for Statistics\\
  University of Bremen\\
  Bremen, Germany
}

\date{March 20, 2020}

\hypersetup{pdftitle={Simultaneous Inference for Multiple Proportions: A Multivariate Beta-Binomial Model}} % Set the PDF's title to your title
\hypersetup{pdfauthor={Max Westphal}} % Set the PDF's author to your name
\hypersetup{pdfkeywords={Bayesian inference, binary data, copula model, credible region, Dirichlet distribution, shrinkage estimation}} % Set the PDF's keywords to your keywords
\hypersetup{
	colorlinks,
	breaklinks=true,
	linkcolor={red!50!black},
	citecolor={blue!50!black},
	urlcolor={red!50!black}
}

\begin{document}
\maketitle

\begin{abstract}
Statistical inference in high-dimensional settings is challenging when standard unregularized methods are employed. In this work, we focus on the case of multiple correlated proportions for which we develop a Bayesian inference framework.
For this purpose, we construct an $m$-dimensional Beta distribution from a $2^m$-dimensional Dirichlet distribution, building on work by \cite{olkin2015}. 
This readily leads to a multivariate Beta-binomial model for which simple update rules from the common Dirichlet-multinomial model can be adopted. From the frequentist perspective, this approach amounts to adding pseudo-observations to the data and allows a joint shrinkage estimation of mean vector and covariance matrix.
For higher dimensions ($m>10$), the extensive model based on $2^m$ parameters starts to become numerically infeasible. To counter this problem, we utilize a reduced parametrisation which has only $1+m(m+1)/2$ parameters describing first and second order moments. 
A copula model can then be used to approximate the (posterior) multivariate Beta distribution. 
A natural inference goal is the construction of multivariate credible regions.
The properties of different credible regions are assessed in 
a simulation study in the context of investigating the accuracy of multiple binary classifiers. It is shown that the extensive and copula approach lead to a (Bayes) coverage probability very close to the target level. In this regard, they outperform credible regions based on a normal approximation of the posterior distribution, in particular for small sample sizes. Additionally, they always lead to credible regions which lie entirely in the parameter space which is not the case when the normal approximation is used.
\end{abstract}

\keywords{Bayesian inference \and binary data \and copula model \and credible region \and Dirichlet distribution \and shrinkage estimation}

\section{Introduction}

This work is motivated by the goal to conduct Bayesian inference for multiple  proportions $\vartheta_j\in (0,1)$, $j=1,\ldots,m$, which are possibly correlated. In particular, the goal is to derive a multidimensional credible region for $\bm \vartheta = (\vartheta_1,\ldots, \vartheta_m) \in (0,1)^m$ taking into account the dependency structure. The author's main application of interest are model evaluation and comparison studies where the accuracy of $m$ binary classifiers is assessed on the same dataset. Previous work in the machine learning context has shown that evaluating multiple promising models on the final test data (instead of a single prespecified model) can improve the final model performance and statistical power in the evaluation study. Hereby, it is beneficial to take into account model similarity when adjusting for multiplicity as the according adjustment needs to be less strict when different models give similar predictions \citep{EOMPM1, IMS}. 

\cite{EOMPM1, IMS} focused was on frequentist methods, in particular a multivariate normal approximation %for the vector of test statistics 
in conjunction with the so called maxT-approach (projection method) \citep{hothorn2008}. From a Bayesian viewpoint, this can (at least numerically) be seen as a multivariate normal-normal model under the assumption of a flat (improper) prior distribution. This approach showed good performance in terms of family-wise error rate control in extensive simulation studies. However, it does have some drawbacks. Firstly, the resulting confidence region is not guaranteed to lie entirely in the parameter space $(0,1)^m$. Secondly, the needed estimate of the covariance matrix may be singular, which is e.g. the case if any observed proportion $\thetahat_j$ is zero or one. This scenario becomes increasingly likely when $m \rightarrow \infty$. Finally, for (close to) least favourable parameter configurations, this approach leads to an increased type 1 error rate in the frequentist sense.
While many ad hoc remedies exist for these problems (e.g. parameter transformation, shrinkage estimation), the main goal of this work to derive a more self-contained model. Moreover, it may be desirable to include prior knowledge to the inference task which calls for a Bayesian treatment of the problem.

For a single proportion $\vartheta \in (0,1)$, a common Bayesian approach is the Beta-binomial model where each observation $X_i$ is a Bernoulli variable such that 
\begin{align}
Y=\sum_i X_{i} \sim \Bin(n, \vartheta).
\end{align}
Assuming a Beta prior distribution with shape parameters $\alpha > 0, \beta > 0$ for $\vartheta$, i.e.
\begin{align}\label{uni_beta_prior}
\vartheta \sim \Beta(\alpha, \beta)
\end{align}
leads to the posterior distribution 
\begin{align}\label{uni_beta_posterior}
\vartheta\ | \ y \sim \Beta(\alpha+y, \beta+n-y)
\end{align}
given $y=\sum_i x_i \in \{0,1,\ldots,n\}$ successes have been observed \citep[pp.\ 172-173]{ASI}. 

The goal of this work is twofold. Firstly, the Beta distribution shall be generalized to higher dimensions allowing general correlation structures. 
A multivariate generalization of the Beta distribution has been studied in several works \citep{jones2002, olkin2003, nadarajah2005, arnold2011}, mostly however limited to the bivariate case. \citet{olkin2015} discuss shortcomings of some of these approaches such as a restricted range of possible correlations or a complicated extension to higher dimensions. Another general access to this problem the separation of marginal distributions and dependency structure via copula models \citep{CMD, CBD, nyaga2017}. The second aim is to then derive a multivariate Beta-binomial model which allows to conduct Bayesian inference regarding $\bm \vartheta$ or transformations thereof. 

To this end, the bivariate Beta distribution by \cite{olkin2015} based on an underlying Dirichlet distribution is extended to higher dimensions. This was already adumbrated in section 3 of the original article. While this construction works in theory for any dimension $m$, it suffers from the fact that it depends on $2^m$ parameters. This may serve problems, in particular when posterior samples consisting of a large number of observations of (initially) $2^m$ variables need to be drawn. In practice, it is thus only feasible for dimensions not much larger than $m=10$ (depending on computational resources). However, it will be shown that a reduced parametrisation with $1+m(m+1)/2$ parameters allows to handle much higher dimensions with reasonable computational effort by employing a copula model.

The construction shown in this article has methodological similarities to existing work in the context of multivariate Bernoulli or Binomial distributions with general correlation structures \citep{madsen1993, kadane2016, fontana2018}. In particular, the question which correlation structures are admissible for an $m$-dimensional Beta distribution can directly be transferred to the same question regarding an $m$-dimensional Bernoulli distribution \citep{hailperin1965, chaganty2006}. Similar considerations have been made regarding the question how to generate correlated binary data \citep{leisch1998, xue2010, preisser2014, shults2017}. The distinctive feature of this work is thus the different (Bayesian) setting and the focus on statistical inference. Moreover, to the best of the author's knowledge, the necessary and sufficient moment conditions that are provided in section \ref{mBeta} have not been mentioned in this form in the literature despite the fact that many authors have recognized the connection to linear programming \citep{madsen1993, fontana2018, shults2017}.

The remainder of this work is structured as follows: In the section \ref{theory}, the construction and several properties of the multivariate Beta distribution as well as restrictions on the admissible correlation structures are described. Moreover, a multivariate analogue to the common Beta-binomial model is introduced. In section \ref{inference}, different methods for the derivation of multivariate credible regions are described. Details regarding the numerical implementation are given in section \ref{software}. Section \ref{sim} covers numerical simulations in the context of model evaluation and comparison studies to access the properties of different credible regions. Finally, section \ref{discussion} contains a summary of the present work and a discussion on the connection between multivariate credible regions and Bayesian hypothesis testing.

\section{Statistical model and theoretical results}\label{theory}

\subsection{Multivariate Beta distribution}\label{mBeta}

Our goal is a joint model for the success probabilities $\vartheta_j$ of $m\in \mathbb{N}$ Binomial variables $Y_j = \sum_i X_{ij} \sim \Bin(n, \vartheta_j)$, $j=1,\ldots, m$, with arbitrary correlation structure between the variables $X_j, X_{j'}$. Conditional on the observed data $Y_j = y_j$, we assume that marginally 
\begin{align}
\vartheta_j\ |\ y_j \sim \Beta(\alpha_j + y_j, \beta_j +n - y_j)
\end{align}
as introduced in equation \eqref{uni_beta_posterior}.
The variables $Y_j$ are the sum $Y_j = \sum_{i=1}^n X_{ij}$ of Bernoulli variables ${X_{ij} \stackrel{iid}{\sim} \Bin(1, \vartheta_j)}$, $j=1,\ldots,m$, which are also observed.

The subsequent construction of an $m$-dimensional Beta distribution 
is based on a $2^m$-dimensional Dirichlet distribution such as proposed by \cite{olkin2015} for the bivariate case. The Dirichlet distribution is frequently employed in the so called Dirichlet-multinomial model for the success probabilities of multinomial data. A multinomial random variable is the generalization of a Binomial random variable, i.e. each observation is one of $w$ distinct events $\{1,\ldots,w\}$ where each event $k$ has a probability of $p_k$ to occur, such that ${||\bm p||_1 = \sum_k p_k = 1}$. A Dirichlet random variable $\bm p = (p_1,\ldots,p_w)\tra \sim \Dir(\bm \gamma)$ has support $\bm \setP= \{\bm p \in (0,1)^w: ||\bm p||_1=1\}$ and is fully characterized by the concentration parameter (vector) $\bm \gamma = (\gamma_1,\ldots,\gamma_w)\tra  \in \mathbb{R}^w_+$. 
A comprehensive overview of the Dirichlet distribution is given by \cite{DIR}.
In the following, an $m$-dimensional random variable with multivariate Beta distribution will be constructed from a $2^m$-dimensional Dirichlet random variable.
We will see that this can be achieved by a convenient parametrisation and a simple linear transformation. Although the case $m=1$ can easily be recovered, $m \geq 2$ is assumed in the following to avoid laborious case distinctions.

A single binary observation is assumed to be a realization of an $m$-dimensional random variable ${\bm X = (X_1, \ldots, X_m)\tra \in \bm \setX= \{0,1\}^m}$.  The complete experimental data, $n$ i.i.d. observations of $\bm X$, is collected in the rows of the $n \times m$ binary matrix $\Xmat$. We define a categorical random variable $\bm C =h^{-1}(\bm X) \in \bm \setC= {\{\bm c \in \{0,1\}^w:\ ||\bm c||_1=1\}}$, $w=2^m$, which is linked to $\bm X$ via the mapping $h$ which is defined in the following. 
\begin{definition}[Transformation matrix]\label{Hmat}
	Define the linear mapping ${h: \mathbb{R}^w \rightarrow \mathbb{R}^m}$, $\bm z \mapsto \bm H \bm z$, whereby the {$j$-th} column of the transformation matrix \protect{$\bm H = \bm H(m) \in \{0,1\}^{m \times w}$}
	corresponds to the binary representation (of length $m$) of the integer $j-1$, $j=1,\ldots,w$. 
\end{definition}
This definition uniquely defines $\bm H(m)$ for any dimension $m$. For instance, for $m=3$,
\begin{align}\label{H3}
\bm H= \bm H(3)=\left( \begin{array}{rrrrrrrr}
0 & 0 & 0 & 0 & 1 & 1 & 1 & 1 \\
0 & 0 & 1 & 1 & 0 & 0 & 1 & 1 \\
0 & 1 & 0 & 1 & 0 & 1 & 0 & 1 
\end{array}\right).
\end{align}
It is easy to see that $h(\bm \setC) = \bm \setX$ and that $|\bm \setX|=|\bm \setC|=2^m$. In effect, $h$ defines a bijection between $\bm \setC$ and $\bm \setX$ which is illustrated below for $m=3$:
\begin{align}
\bm X = (0,0,0)\tra \quad  &\Leftrightarrow \quad \bm C = %\bm c_1 =
(1, 0, 0, 0, 0, 0, 0, 0)\tra \\
\bm X = (0,0,1)\tra \quad  &\Leftrightarrow \quad \bm C = %\bm c_2 =
(0, 1, 0, 0, 0, 0, 0, 0)\tra \\
\bm X = (1,1,0)\tra \quad  &\Leftrightarrow \quad \bm C = %\bm c_3 = 
(0, 0, 1, 0, 0, 0, 0, 0)\tra \\
&\ \ \vdots \\
\bm X = (1,1,0)\tra \quad  &\Leftrightarrow \quad \bm C = %\bm c_7 = 
(0, 0, 0, 0, 0, 0, 1, 0)\tra\\
\bm X = (1,1,1)\tra \quad  &\Leftrightarrow \quad \bm C = %\bm c_8 = 
(0, 0, 0, 0, 0, 0, 0, 1)\tra. 
\end{align}

Hence, the function $h$ defines a one-to-one correspondence between observing (a) $m$ correlated Bernoulli variables $X_j$ and (b) a single categorical variable $\bm C$ with $2^m$ possible events. The same link may be used to relate (a) $m$ correlated Binomial variables $Y_j=\sum_i X_{ij}$ and (b) a single multinomial variable $\bm D=\sum_i \bm C_{i}$. 
A realization $\bm d=\bm d(\Xmat) = \bm d(\bm \Cmat)$ of $\bm D$ is referred to as the cell count version of the experimental data. It can easily be computed as the sum of all rows of the matrix $\Cmat$. %As $\Xmat= \Cmat \bm H$ we may equally write $\bm d = \bm d(\Xmat)$.
Clearly, the probabilities 
%\begin{align}
$p_k = \pr(\bm C = \bm c_k)$
%\end{align}
for the $w$ distinct events $\bm c_k$, $k = 1,\ldots,w$, can be modelled via the Dirichlet distribution as $||\bm p||_1=1$. This allows us to define the random variable $\bm \vartheta = \bm H \bm p$ and investigate it's properties.

\begin{definition}[Multivariate Beta (mBeta) distribution]\label{mBeta_def}%\leavevmode
	Let $m\geq 2$, $w=2^m$ and ${\bm \gamma \in \R_+^w}$. Let $\bm p = (p_1,\ldots,p_w)\tra \sim \Dir(\bm \gamma)$ follow the Dirichlet distribution with concentration parameter $\bm \gamma$. % where $p_k = \pr(\bm C = \bm c_k)$ and $\bm C \in \bm \setC$ as defined previously such that .... %Then $\bm X = \bm H \bm C$ and we can
	Define the linear transform $\bm \vartheta = \bm H \bm p$ of $\bm p$ whereby the transformation matrix $\bm H = \bm H(m) \in \{0,1\}^{m \times w}$ is defined in definition \ref{Hmat}. In this case $\bm \vartheta$ is said to follow a multivariate Beta distribution with concentration parameter $\bm \gamma$ or $\bm \vartheta \sim \mBeta(\bm \gamma)$ for short.
\end{definition}

\begin{prop}[Properties of the mBeta distribution]\label{properties}
	Let $\bm \vartheta \sim \mBeta(\bm \gamma)$ as defined in definition \ref{mBeta_def}. Then the following assertions hold:
	\begin{enumerate}
		\item $\vartheta_j = \pr (X_j = 1)$ for $j=1,\ldots,m$.
		\item $\bm \vartheta$ is a $m$-dimensional random variable with support $\bm \Theta = (0,1)^m$.
		\item For any $j \in \{1,\ldots,m\}$, $\vartheta_j$ marginally has a $\Beta(\alpha_j, \beta_j)$ distribution with parameters ${\bm \alpha = \bm H \bm \gamma}$ and $\bm \beta = \nu - \bm \alpha$ whereby $\nu = ||\bm \gamma ||_1$. 
		\item The mean vector of $\bm \vartheta$ is given by $\E(\bm \vartheta) = \bm \alpha/\nu$.
		% $\bm A = \bm H (\bm H\tra \odot \bm \Gamma)$
		\item Define $\bm \Gamma = \diag(\bm \gamma)$ and $\bm A = \bm H \bm \Gamma \bm H\tra \in \mathbb{R}^{m \times m}_+$. Then $\bm \alpha = \diag(\bm A)$ and 
		the covariance of $\bm \vartheta$ is given by 
		\begin{align}
		\cov(\bm \vartheta)  
		= \bm \Sigma = \left(\nu \bm A -  \bm \alpha \bm \alpha\tra \right)/(\nu^2(\nu+1)).
		\end{align}
		\item Probabilities of products  $\vartheta_J%=\prod_{j\in J}\vartheta_j 
		= \pr( \bigcap_{j \in J} \{X_j=1\}) = \pr(\prod_{j \in J} X_j = 1) $ with $J \subset \{1,\ldots,m\}$ have a  $\Beta(\alpha_J, \beta_J)$ distribution with
		\begin{align}
		\alpha_J= \bm H_J\bm \gamma \in \R \quad \text{and} \quad \beta_J= \nu - \alpha_{J}.
		\end{align}
		Hereby, $\bm H_J= (\odot_{j \in J} \bm H_{j:}) \in \{0,1\}^{1 \times w}$ is the Hadamard product of associated rows $\bm H_{j:}$ of $\bm H$.
	\end{enumerate}
\end{prop}

Most of the above claims follow immediately from the definition of $\bm \vartheta$. Further details are provided in appendix \ref{proofs}. The symmetric matrix $\bm A$ contains the (scaled) first-order moments $\bm \alpha = \nu \bm \mu = \diag(\bm A)$ and mixed second-order moments $\alpha_{jj'}$ as off-diagonal elements and will prove to be useful later. \citet{olkin2015} state that the density function of $\bm \vartheta$ does not have a closed form expression and show several representations for different subregions of the unit square in the bivariate setting.
The next definition will allow a characterization of the correlation structures that are admissible for a multivariate Beta distribution.

\begin{definition}[Moment conditions]\label{mc_def}
	Let $\nu \in \R_+$ and $\bm A \in \R_+^{m \times m} $ be a symmetric matrix. Define ${\bm 1_w = (1,\ldots,1)\tra \in \R^w}$,
	\quad
	\begin{align}
	\bm H^{(2)} &= \left(\bm H_{j:} \odot \bm H_{j':}\right)_{\substack{\ j=1,\ldots,m-1\\ j'=j+1,\ldots,m}} = \begin{pmatrix} 
	\quad \ \bm H_{1:} \odot \bm H_{2:}\\
	\quad \ \bm H_{1:} \odot \bm H_{3:}\\
	\quad \ \vdots \\
	\bm H_{(m-1):} \odot \bm H_{m:}\\
	\end{pmatrix} \quad\\
	%{j=1,\ldots,m\\ j\tra } %h(\bm c_1) \odot h(\bm c_2) ???
	\text{and}\quad  \widetilde{\bm H} &= \begin{pmatrix} 
	\bm H\\
	\bm H^{(2)}\\
	\bm 1_w\tra
	\end{pmatrix} \in \{0,1\}^{r \times w}
	\end{align}
	with $r=1+m(m+1)/2$. Hereby $\bm H_{j:}$ is the $j$-th row of $\bm H$ and $\odot$ the Hadamard (entrywise) product of vectors. In addition, let $\widetilde{\bm \alpha} = (\bm \alpha\tra, \bm \alpha^{(2)\top}, \nu)\tra $ whereby $\bm \alpha^{(2)}=(\alpha_{12},\alpha_{13},\ldots, \alpha_{(m-1)m})\tra$ contains the upper off-diagonal elements of $\bm A$. Then the pair $(\nu, \bm A)$ is said to satisfy the moment conditions if
	\begin{align}\label{mc}\tag{MC}
	\forall \bm b \in \R^{1+m(m+1)/2}:\quad  \widetilde{\bm H}\tra \bm b \geq \bm 0 \ \Rightarrow\  \bm b\tra \widetilde{\bm \alpha}  \geq 0.
	\end{align}
\end{definition}

\quad

Note that the binary matrix $\widetilde{\bm H}$ only depends on the dimension $m$ and not on $\bm A$. By imputing suitable vectors $\bm b$ in \eqref{mc} it is easy to see that the moment bounds
\begin{align}\label{mb}\tag{MB}
\forall J \subset \{1,\ldots,m\}:\quad \nu \geq \sum_{j \in J} \alpha_j - \sum_{j,j' \in J: j \neq j'} \alpha_{jj'}
\end{align}
are a consequence of \eqref{mc}. Hereby, the sum over an empty index set is defined to be zero.
Moreover, \eqref{mc} also implies the Fr\'echet type bounds
\begin{align}\label{fb}\tag{FB}
\max(\bm 0,\ \bm R_{\bm A} + \bm C_{\bm A} - \nu) = \bm A ^- \leq \bm A \leq \bm A^+ =\ \min(\bm R_{\bm A}, \bm C_{\bm A})
\end{align}
whereby both inequalities and min and max operations are meant component-wise. Hereby $\bm R_{\bm A}$ is the $m \times m$ matrix with all rows identical and equal to ${\bm \alpha = \diag(\bm A)}$ and $\bm C_{\bm A} = \bm R_{\bm A}\tra$. The derived conditions \eqref{mb} and \eqref{fb} or variations thereof have appeared several times in the relevant literature \citep{leisch1998, shults2017}.

\begin{prop}[mBeta parametrisation]\label{main} \leavevmode
	%\quad
	\begin{enumerate}
		\item (Existence) Let $\nu\in \R_+$, $\bm \mu \in (0,1)^m$ and $\bm R \in (-1,1)^m$ a valid (symmetric, positive-definite) correlation matrix with $\diag(\bm R)=\bm 1_m$. Then, there exists a vector $\bm \gamma \in \R_+^w$, $w=2^m$, with $||\bm \gamma||_1=\nu$ and a random variable ${\bm \vartheta \sim \mBeta(\bm \gamma)}$ such that 
		\begin{align}
		\E(\bm \vartheta) = \bm \mu \quad \text{and} \quad \cov(\bm \vartheta) = \bm \Sigma = \bm V^{1/2} \bm R \bm V^{1/2}
		\end{align}
		if and only if $\nu$ and the derived moment matrix
		\begin{align}
		\bm A = \bm A(\nu, \bm \mu, \bm R) = \nu((\nu +1) \bm \Sigma + \bm \mu \bm \mu\tra)
		\end{align}
		satisfy \eqref{mc}. Hereby, $\bm V = \diag(\bm\mu \odot (1 - \bm \mu))/(\nu+1)$. 
		\item (Uniqueness) Given parameters $\nu$ and $\bm A$ fulfilling $\eqref{mc}$ as in (1), the parameter $\bm \gamma$ can be uniquely determined if and only if $m=2$. For ${m>2}$, uniqueness can be achieved by imposing additional constraints, e.g. by minimization of $||\bm \gamma - \bm 1_w\nu/w||_2$. 
	\end{enumerate}	
\end{prop} 
The first result can be proven by applying Farkas' Lemma, a standard result from linear programming, to the linear program 
\begin{align}\label{lp}\tag{LP}
\widetilde{\bm H} \bm \gamma = \widetilde{\bm \alpha} 
\end{align} which needs to be solved for $\bm \gamma$. Details are given in appendix \ref{proofs}.

The moment conditions give some intuition on the admissible correlation structures, in particular by means of the weaker but more interpretable necessary conditions \eqref{mb} and \eqref{fb}. However, a direct verification of \eqref{mc} is usually not feasible in practice,
at least not more efficiently than attempting to solve \eqref{lp}. A more practical approach to translate a correlation into a moment description is described in section \ref{software}.

The bounds on the derived moment matrix $\bm A$ induced by \eqref{mc} imply bounds on the correlation matrix $\bm R$ because the elements of $\bm A$ are monotone in the according elements of $\bm R$. The construction and the according conditions translate to a multivariate Bernoulli distribution with minor modifications. %, primarily $\nu=1$ has to be fixed.
In this context, several works have illustrated the %\eqref{mb}
bounds on the correlation coefficients for low dimensions \citep{prentice1988, preisser2014, shults2017}. For instance, for $m=2$, \eqref{fb} implies
\begin{align}
\max \left( - (\psi_1\psi_2)^{-1} %\frac{-1}{\psi_1\psi_2}
, -\psi_1 \psi_2 \right)\leq \rho_{12} \leq \min \left(  \frac{\psi_1}{\psi_2}, \frac{\psi_2}{\psi_1} \right)
\end{align}
with $\psi_j = \sqrt{\mu_j/(1-\mu_j)} $, $j=1,2$.
These necessary correlation bounds also apply to the case ${m>2}$ for all $\rho_{jj'}$ but are only sufficient for $\bm R$ being admissible for $m=2$.
It should be noted, that the overall concentration parameter $\nu=||\bm \gamma||_1$ only drives the variances of the $\vartheta_j$. That is to say, two mBeta distributions induced by the parameters $\bm \gamma_1$ and $\bm \gamma_2$ with $\bm \gamma_1 / ||\bm \gamma_1||_1 = \bm \gamma_2 / ||\bm \gamma_2||_1 $ have the same correlation structure. Below, several simple results concerning \eqref{mc} are described, some of which will be utilized in the next section. 

\begin{prop}[\eqref{mc} in practice]\label{practice}
	\quad
	\begin{enumerate}%[label=(\alph*)]
		\item For all $\bm \gamma \in \R_+^w$, the pair $\nu=||\bm \gamma||_1$,  $\bm A = \bm H \bm \Gamma \bm H \tra$ satisfies \eqref{mc}. 
		\item Let $\bm d = \bm d(\bm \X)$ be the cell count version of the experimental data and $\bm \Delta = \diag(\bm d)$. Then the pair $n=||\bm d||_1$, $\bm U = \bm H \bm \Delta \bm H$ satisfies \eqref{mc}.
		\item If $(\nu, \bm A)$ and $(n, \bm U)$ both satisfy \eqref{mc}, the pair $(\nu^*, \bm A^*)=(\nu+ n, \bm A + \bm U)$ does as well.
		\item Not all $\nu \in \R_+$ and $\bm A  \times \R_+^{m \times m}$ satisfy \eqref{mc}.
	\end{enumerate}
\end{prop}

\subsection{Multivariate Beta-binomial model}\label{model}

The next result formalizes that the update rule $\bm \gamma^* = \bm \gamma + \bm d$ from the well-studied Dirichlet-multinomial model can be adopted for the multivariate generalization of the Beta-binomial model. Hereby, observed cell counts $d_k$ are added to corresponding prior parameters $\gamma_k$.

\begin{prop}[Multivariate Beta-binomial model]\label{mBeta_model}
	Let $\bm \vartheta \sim \mBeta(\bm \gamma)$ be the prior distribution for $\bm \vartheta$ and let $\Xmat \in \{0,1\}^{n \times m}$ be the observed data matrix with cell count representation $\bm d = \bm d(\Xmat)$. Then:
	\begin{enumerate}
		\item The posterior distribution of $\bm \vartheta$ is given by
		\begin{align}
		\bm \vartheta\given \bm d \sim \mBeta(\bm \gamma^*)
		\end{align} %$$ 
		whereby  $\bm \gamma^* = \bm \gamma + \bm d$.
		\item Let $\bm \gamma$ and $\bm \gamma^*$ be the parameter of prior and posterior distribution of $\bm \vartheta$, respectively. Let ${\bm A = \bm H \bm \Gamma \bm H\tra}$ and $\bm A^* = \bm H \bm \Gamma^* \bm H\tra$. Then
		\begin{align}
		\bm A^* = \bm A + \bm U \quad \text{and} \quad \nu^* = \nu + n
		\end{align}
		whereby the update matrix $\bm U = \bm H \bm \Delta \bm H\tra $ depends on the data $\bm \X$ due to $\bm \Delta = %\bm \Delta(\Xmat) = 
		\diag(\bm d(\bm \X))$.
	\end{enumerate}
	
\end{prop}

The second result is useful as it allows to work with a reduced parametrisation from which the mean vector and covariance matrix of the distribution can still be derived, see proposition \ref{properties}. It depends on $\nu \in \mathbb{R}_+$ and the symmetric matrix $\bm A \in \mathbb{R}_+^{m \times m}$ and thus requires $1+(m+1)m/2$ parameters. Hereby, neither the prior parameter $\bm \gamma$ nor the posterior parameter $\bm \gamma^*$ are needed to derive the posterior matrix $\bm A^*$ from the prior matrix $\bm A$ and the data $\Xmat$. The term 'reduced parametrisation' will be used when $\nu$ and $\bm A$ are known but $\bm \gamma$ is unknown. 

\begin{figure*}%[ht!]
	%\label{viz}
	\begin{subfigure}{0.49\textwidth}
		\includegraphics[width=1\linewidth]{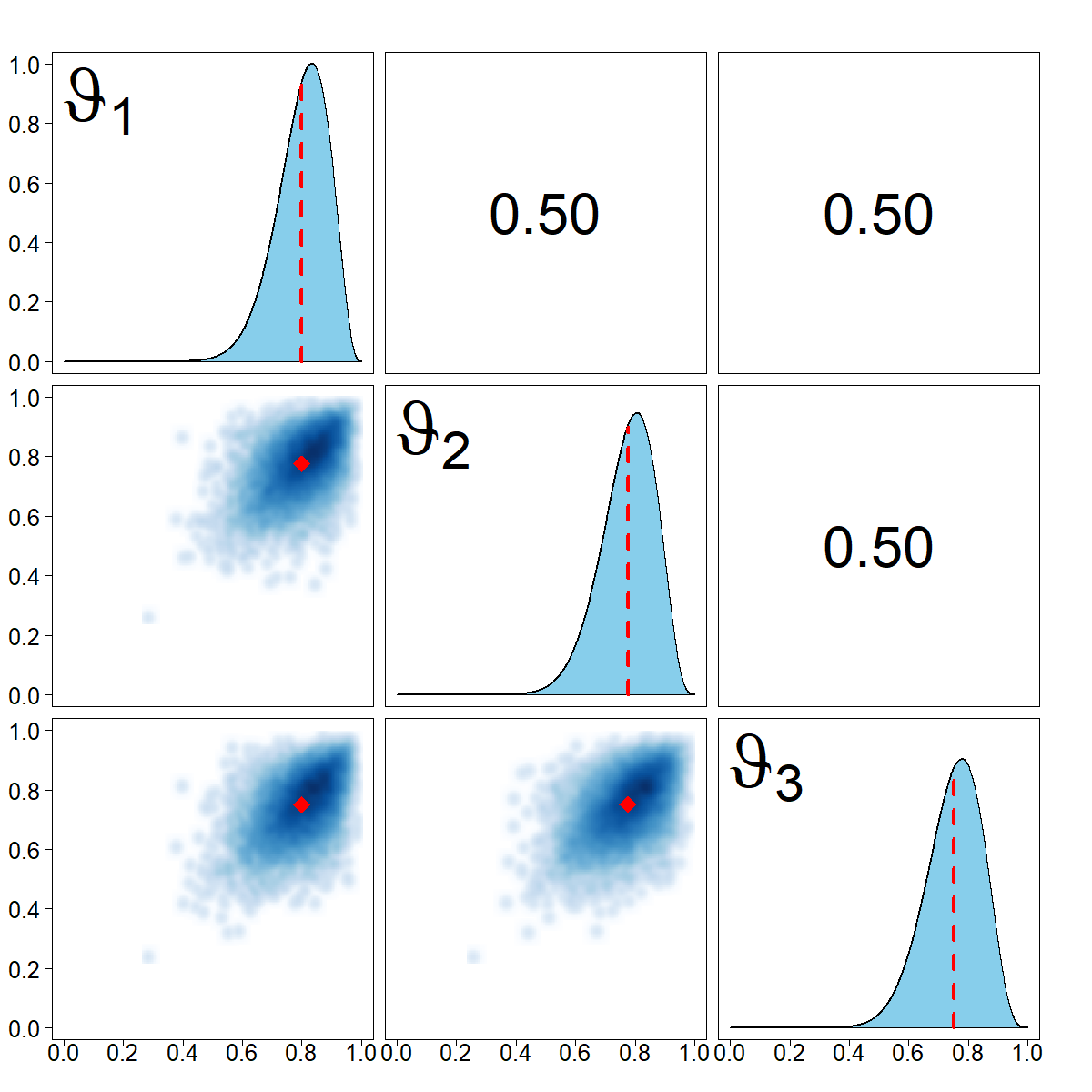}
		\subcaption{Prior distribution ($\nu =20$)}
	\end{subfigure}
	\hfill	
	\begin{subfigure}{0.49\textwidth}
		\includegraphics[width=1\linewidth]{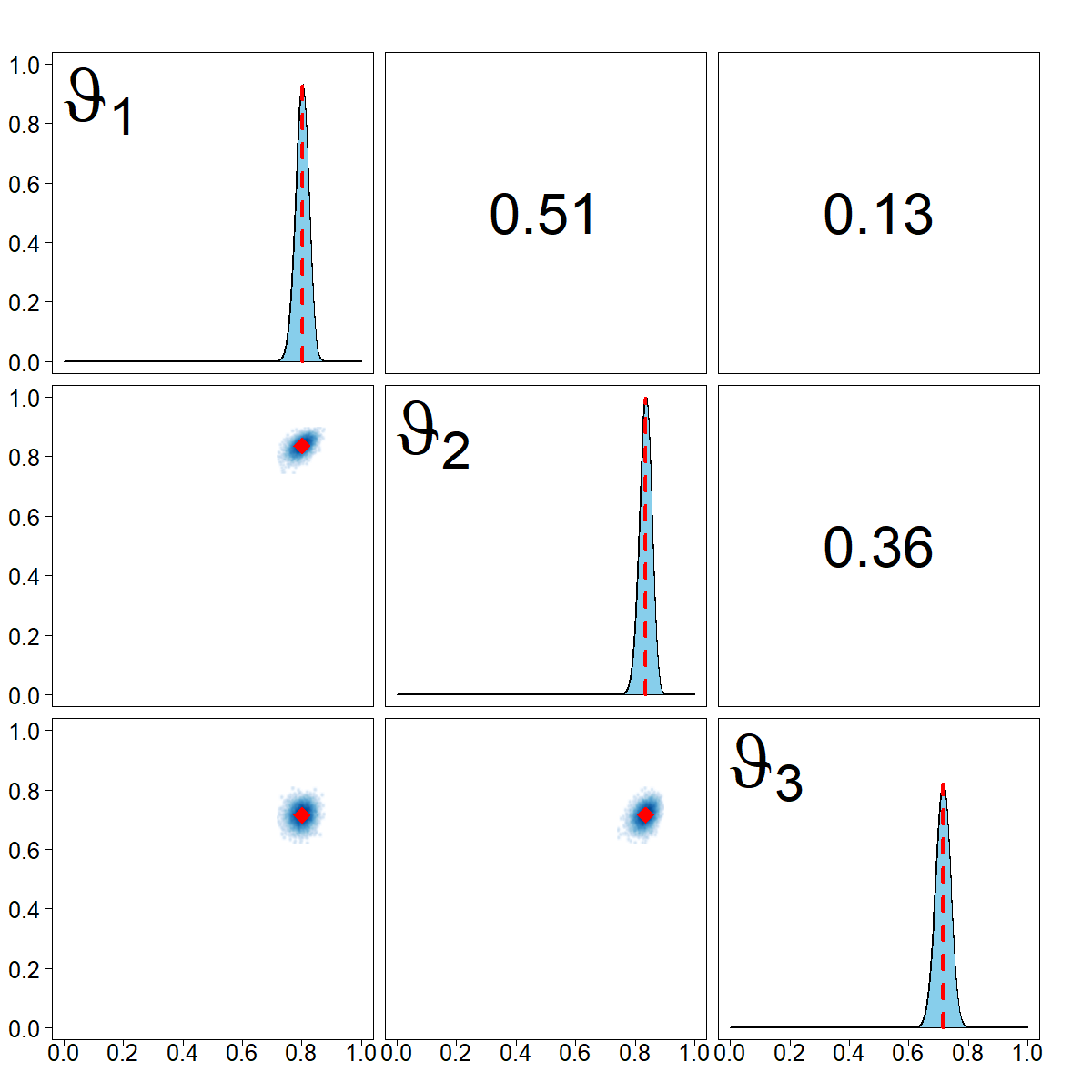}
		\subcaption{Posterior distribution ($\nu^*=337$)}
	\end{subfigure}
	\caption{Visualization of prior and posterior mBeta distribution corresponding to the example in appendix \ref{example}. Plots show the marginal densities (diagonal), bivariate densities (lower panel) and correlation coefficients (upper panel).}
	\label{figure1}
\end{figure*}

When the reduced parametrisation is employed, as a consequence of proposition \ref{practice}, the only worry is to correctly specify a prior parameter $\bm A$, either directly or implicitly (via $\bm R$). As previously stated, checking \eqref{mc} for the prior distribution may not always be feasible, especially in high dimensions. For certain priors with simple structure, this is however easily possible. In particular, mBeta distributions with $\nu \in \R_+$, $\bm \mu= \bm 1_m /2$ and $\bm R = \bm I_m$, i.e. independent $\Beta(\nu/2, \nu/2)$ distributions are always admissible. One can simply check that one possible parameter vector to obtain these properties is $\bm \gamma = \bm 1_w \nu / w$. The case $\nu = 2$ corresponds to independent uniform distributions over $(0,1)^m$ which will be employed as a vague prior in the simulation study in section \ref{sim}. Another simple and practically relevant way to ensure the validity of $\bm A$ is to construct it based on previous experimental data $\bm d_p$ via $\bm A=\bm H \diag(\bm d_p) \bm H\tra$, see proposition \ref{practice}. Figure \ref{figure1} illustrates the update rule by visualizing a prior and posterior distribution in the three dimensional case. The underlying numerical example is provided in appendix \ref{example}.

\section{Bayesian inference}\label{inference}

\subsection{Construction of credible regions}

Assume that the prior distribution $\pi \equiv \mBeta(\bm \gamma)$ has been updated to the posterior distribution $\pi^* \equiv \mBeta(\bm \gamma^*)$ with $\bm \gamma^* = \bm \gamma + \bm d$ (proposition \ref{mBeta_model}). Deriving a simultaneous credible region for all proportions $\bm \vartheta \given \bm d \sim \pi^* $ of interest is a typical data analysis goal. A $100(1-\alpha)$\% credible region is a set {$\CR_{1-\alpha} \subset (0,1)^m$} with the property $\pr_{\pi^*}(\bm \vartheta \in \CR_{1-\alpha}%\ |\ \bm \gamma^*
) = 1-\alpha$ \citep{SDT}.
For simplicity, only equi-tailed, two-sided credible regions are considered in this work. In the normal posterior case, this corresponds to highest density regions (HDR) while for the mBeta distribution this is in general not true. 

There are several ways to construct credible regions for $\bm \vartheta$. One possibility is to approximate the posterior distribution by a multivariate normal by matching the first two moments and use established methods for this case. 
The normal approximation of this posterior distribution is then given by 
\begin{align}
\bm \vartheta \given \bm d \stackrel{\cdot}{\sim} \setN_m(\bm \mu^*, \bm \Sigma^*)
\end{align}
whereby $\bm \mu^*$ and $\bm \Sigma^*$ are derived according to proposition \ref{properties} from $\bm A^*$. From this, a simultaneous credible region 
\begin{align}
\CR_{1-\alpha} =  \prod_{j = 1}^m\left( \mu^*_j - c_\alpha ( v^*_j)^{1/2},\  \mu^*_j + c_\alpha  (v^*_j)^{1/2}   \right),
\end{align}
can be derived. Hereby, the posterior variance vector $\bm v^* \in \R^m$ contains the diagonal elements of $\bm \Sigma^*$. The 'critical' constant $c_\alpha$ can be computed numerically as a suitable equi-tailed quantile of the standard multivariate normal with correlation matrix $\bm R^* = \diag(\bm v^*)^{-1/2} \bm \Sigma \diag(\bm v^*)^{-1/2} $ \citep{hothorn2008}.
Note that such approximate credible regions are not guaranteed to lie in $(0,1)^m$. Moreover, this approach does not benefit from the fact that the type of the marginal posterior distributions is known and non-Gaussian.

This can be alleviated by employing a copula approach \citep{nadarajah2018}. A copula model allows to disentangle marginal distributions $F_{X_1},\ldots,F_{X_m}$ and dependency structure of a multivariate random variable $\bm X =(X_1,\ldots,X_m)\tra$. More specifically, Sklar's theorem states that for every random vector $\bm X$ with joint cumulative distribution function (CDF) $F_{\bm X}$ there exists a copula function $C:[0,1]^m \rightarrow [0,1]$ such that
\begin{align}
F_{\bm X}(\bm x) = C(F_{X_1}(x_1),\ldots,F_{X_m}(x_m)), \quad \bm x \in \bar{\R}^m.
\end{align}
Furthermore, the copula $C$ is unique if all $m$ marginal distributions $F_{X_j}$ are continuous \citep{sklar1959, nadarajah2018}. For instance, a Gaussian copula may be utilized which is parametrized via a correlation matrix $\bm R_m$ and given by
\begin{align}
C_{\bm R_m}(\bm u) = \Phi_m(\Phi^{-1}(u_1),\ldots,\Phi^{-1}(u_m);\, \bm 0_m, \bm R_m).
\end{align}
Hereby $\Phi$ is the univariate standard normal CDF and $\Phi_m$ is the $m$-dimensional normal CDF with mean $\bm 0_m$ and covariance matrix $\bm R_m$. When modelling the posterior distribution $\bm \vartheta\given \bm d$, an obvious choice for $\bm R_m$ is $\bm R^*$, the posterior correlation matrix which can be obtained by standardizing the posterior covariance $\bm \Sigma^*$ (see proposition \ref{properties}). Following similar arguments as given by \citet{dickhaus2012}, this can be used to construct a simultaneous credible region. For this, the same constant $c_\alpha$ as for the normal approximation is used and translated to adjusted local tail probabilities $\tilde{\alpha} = 1-\Phi(c_\alpha)$. The credible region is then based on the $\tilde{\alpha}/2$ and $(1-\tilde{\alpha}/2)$ quantiles of the $m$ marginal Beta distributions of the the joint mBeta posterior.

Lastly, we may base our inference regarding $\bm \vartheta \given \bm d$ on $\bm p \given \bm d$, the underlying Dirichlet-multinomial model (section \ref{theory}).
To pursue this route, a posterior sample $\bm \setS_{\bm p} \in \bm \setP^{n_r} \subset (0,1)^{n_r  \times w} $ of size $n_r$ can be drawn from the underlying $\Dir(\bm \gamma^*)$ Distribution which is then transformed to a posterior sample $\bm \setS_{\bm \vartheta} = \bm  \setS_{\bm p} \bm H\tra \in \bm \Theta^{n_r} \subset (0,1)^{n_r  \times m}$ for $\bm \vartheta$, compare section \ref{mBeta}. Denote by
\begin{align}
\CR (\tilde{\alpha}) = \prod_{j=1}^m (\vartheta^-_j, \vartheta^+_j)
\end{align}
the credible region such that $\pr( \vartheta_j < \vartheta^-_j) = \pr( \vartheta > \vartheta^+_j) = \tilde{\alpha}/2$ is satisfied for all margins $j$, meaning that $\vartheta^-_j$, $\vartheta^+_j$ are suitable marginal Beta quantiles. Subsequently, $\tilde{\alpha}$ can be tuned such that $(1-\alpha)n_r$ data points of $\bm \setS_{\bm \vartheta}$ are contained in $\CR (\tilde{\alpha})$. This can be achieved by a simple numerical root finding. While normal approximation and copula model only require knowledge of the reduced parametrisation $(\nu, \bm A)$, this extensive posterior sampling approach is only feasible when the complete parameter vector $\bm \gamma^*$ is known. It is thus the only of the three approaches which employs all available information - if it is indeed available. In high dimensions it is however computationally expensive or even infeasible as the original Dirichlet sample is of size $n_r \cdot 2^m$.

Note that all three approaches to construct credible regions are Bayes actions in the sense that they are based on (different approximations of) the posterior expected coverage probability. In section \ref{sim}, the influence of these approximations on the Bayes coverage probability, i.e. the expected coverage under different generative prior distributions, will be assessed in a simulation study.

\subsection{Inference for transformed parameters}

Besides inference for the proportions $\vartheta_j$ themselves, transformations of them might also be of interest. The three approaches described in the last section (normal approximation, copula model, extensive sampling approach) can be modified for this purpose. A commonly investigated case are linear contrasts defined by a contrast matrix $\bm K \in \R^{t \times m}$ where $t$ is the dimension of the target space. A popular example are all-vs-one comparisons (w.l.o.g.) defined by $\bm K = \left( \bm I_{m-1}, -\bm 1_{m-1} \right) \in \{-1,0,1\}^{(m-1) \times m}$. Hereby $\bm I_{m-1}$ is the $(m-1)$-dimensional identity matrix and $\bm 1_{(m-1)} = (1,\ldots,1)\tra$. In the model evaluation context, this would relate to comparing the accuracy $\vartheta_j$ of all models $j=1,\ldots,m-1$ against $\vartheta_m$, the accuracy of the $m$-th model.

For the normal approximation and the copula method, the fact that $\bm \vartheta \given \bm d \stackrel{\cdot}{\sim} \bm \setN_m(\bm \mu^*, \bm \Sigma^*) $ implies
\begin{align}
\bm K \bm \vartheta \given \bm d \stackrel{\cdot}{\sim} \bm \setN_t(\bm K \bm \mu^*, \bm K \bm \Sigma^* \bm K\tra)
\end{align}
can be utilized.
For the copula approach, it is important to note, that the difference of two Beta random variables no longer follows a Beta distribution \citep{gupta2004}. One solution to obtain correct marginal quantiles is to rely on posterior sampling for this case as well. Non-linear transformations will not be investigated in this work, could however be tackled by employing the multivariate delta method. %\citep{ASI}.
For the extensive sampling approach, any transformation can be applied to the posterior sample. The transformed sample can then be processed by the same means as before.

\section{Numerical implementation}\label{software}

Functions for prior definition, update rules and calculation of credible regions have been implemented in an \texttt{R} package\footnote{A development version of the \texttt{SIMPle} package is available at \url{https://github.com/maxwestphal/SIMPle} (accessed March 20, 2020).}. It's main goal is to conduct simultaneous inference for multiple proportions by means of the proposed multivariate Beta-binomial model. It allows the definition of a prior distribution based on the mean vector and correlation matrix.
Instead of solving the obvious linear system \eqref{lp}, the least squares problem with equality and inequality constraints  
\begin{align}\label{lsei}\tag{LS}
\min_{\bm \gamma \in \R^w}&||\bm H^{(2)} \bm \gamma - \bm \alpha^{(2)} ||_2^2 \\
\text{subject to} \quad 
&\begin{pmatrix}
\bm H\\
\bm 1_w\tra 
\end{pmatrix} \bm \gamma = 
\begin{pmatrix}
\bm \alpha\\ \nu
\end{pmatrix} \quad \text{and} \quad \bm I_w \bm \gamma \geq \bm 0_w \nonumber
\end{align}
is solved for a given input $(\nu, \bm A)$ with help of the \texttt{lsei} %\footnote{See \url{https://CRAN.R-project.org/package=lsei} (accessed July 07, 2019).} 
package \citep{SLSP, lsei}.
That is to say, we require the first-order moments to be matched exactly and the mixed second-order moments should be fitted as closely as possible. The moment matrix $\bm A$ can be specified explicitly or implicitly via mean vector $\bm \mu$ and correlation matrix $\bm R$, compare proposition \ref{main}. If a solution $\bm \gamma$ of \eqref{lsei} is found, it defines a valid mBeta distribution. If the solution is exact, i.e. ${||\bm H^{(2)} \bm \gamma - \bm \alpha^{(2)} ||_2^2=0}$, it defines an mBeta distribution with exactly the targeted mean and correlation structure. If $||\bm H^{(2)} \bm \gamma - \bm \alpha^{(2)} ||_2^2>0$, the solution $\bm \gamma$ defines a valid mBeta distribution with targeted mean but only approximated correlation structure. 

For dimensions $m>10$ the reduced parametrisation in conjunction with the copula approach described in section \ref{inference} is employed by default because solving \eqref{lsei} or \eqref{lp} becomes numerically expensive. In this case only the necessary bounds \eqref{fb} are checked. As a result, the prior matrix $\bm A$ is in general not guaranteed to satisfy \eqref{mc}, unless simplifying structural assumptions are made.

%%%%%%%%%%%%%%%%%%%%%%%%%%%%%%%%%%%%%%%%%%%%%%%%%%%%%%%%%%%%%%%%%%%%%%%%%%%%%%%%%%%%%%%%%%%%%%%%%%%%%
\section{Simulation study: comparison of multiple classifiers}\label{sim}

\subsection{Method comparison}\label{goal}

This section covers the results of a simulation study in the context of classifier evaluation. In machine learning, prediction models should be trained and evaluated on independent data sets to avoid an overoptimistic performance assessment \citep{ESL, ELA, APM}. Earlier work has shown that a simultaneous evaluation of multiple promising classifiers is beneficial as the test data can then be employed for the final model selection. \citep{EOMPM1, IMS}. Hereby, an adequate adjustment for the introduced selection-induced bias needs to be employed. 

The goal of this simulation study is to compare the properties of different credible regions which have been outlined in section \ref{inference}:
\begin{enumerate}
	\item \textbf{approximate}: %multivariate 
	normal approximation of posterior distribution
	\item \textbf{copula}: exact posterior marginals, copula model for dependency structure
	\item \textbf{extensive}: based on a posterior sample of size $n_r=10,000$, drawn from the underlying Dirichlet distribution 
\end{enumerate}

Our primary interest is to assess the Bayes coverage probability of these credible regions, i.e. the expected coverage probability 
\begin{align}
\BCP  = \E_{\pi_g} \one (\bm \vartheta \in \CR_{1-\alpha}),
\end{align}
when parameters arise from different generative prior distributions $\pi_g$. The BCP definition is inspired by the standard Bayes risk definition \citep[p.\,11]{SDT}. Each credible region depends on the employed approach and on the analysis prior $\pi = \pi_a$. We investigate two cases here: (a) $\pi_a = \pi_g$, i.e. the true generative prior is known, and (b) a vague prior is used. The latter case will implemented as $m$ independent uniform variables, corresponding to the parameter $\bm \gamma_a = {\bm 1_w \cdot 2/w}$ which was discussed at the end of section \ref{model}. Because all approaches are constructed as (approximate) Bayes actions, we expect a $\BCP$ close to $1-\alpha$ when (a) $\pi_a = \pi_g$ or (b) the sample size is large.

\subsection{Scenarios}\label{setup}

The employed generative prior distributions $\pi_g \equiv \mBeta(\bm \gamma_g)$ are characterized by dimension $m$, concentration parameter $\nu$, mean vector $\bm \mu$ and correlation matrix $\bm R$. For a given scenario $g=(m, \nu, \bm \mu, \bm R)$ the according parameter vector $\bm \gamma_g$ is obtained by solving \eqref{lsei}. The following cases are investigated: for $m=5$, we consider means of $\mu_j=0.75$ for all models with a concentration parameter of $\nu=20$ or $40$ and a equicorrelation of $\rho =0.5$ or $0.75$. For $m=10$, we define two blocks of five models as above. The prior mean is given as $\bm \mu=(0.75,\ldots,0.75, 0.7,\ldots, 0.7)$. This is supposed to mimic the case that two learning algorithms with different hyperparameters are investigated whereby one of them (averaged over the hyperparameters) yields classifiers with higher accuracy. The correlation between models of different algorithms is defined to be $\rho^2$ such that the overall correlation matrix is a 
block matrix consisting of $5\times 5$ blocks. 

For each simulation run, the underlying parameter vector $\bm p$ is drawn from a $\Dir(\bm \gamma_g)$ distribution. The experimental data is then drawn from a $w-$dimensional multinomial distribution with parameters $n, \bm p$ and then transformed to a multivariate Binomial distribution with the same parameters via $\bm X = \bm H \bm C$, compare section \ref{mBeta}. We have investigated the sample sizes $n = 50, 100, 200, 400, 800$.  

As the BCP is a proportion, the standard error of its simulated estimate is bounded from above by $0.5/\sqrt{N_{sim}}$ which is approximately 0.001 in the overall analysis (figures \ref{figure2} and \ref{figure3}; $N_{sim}=200,000$) and 0.002 in the stratified analysis (appendix \ref{details}; $N_{sim}=50,000$). The three investigated methods are applied to the same simulated datasets. The target coverage probability is set to $1-\alpha=0.95$ for all simulations.
The numerical experiments were conducted in \texttt{R} with help of the \texttt{batchtools} package \citep{batchtools}. Software and custom functions that were used to conduct the simulation study are partially publicly available.\footnote{Compare {\url{https://github.com/maxwestphal/SEPM.PUB}} (accessed March 20, 2020).}

\subsection{Results}

\begin{figure*}[t!]
	\centerline{%
		\includegraphics[width=.9\linewidth]{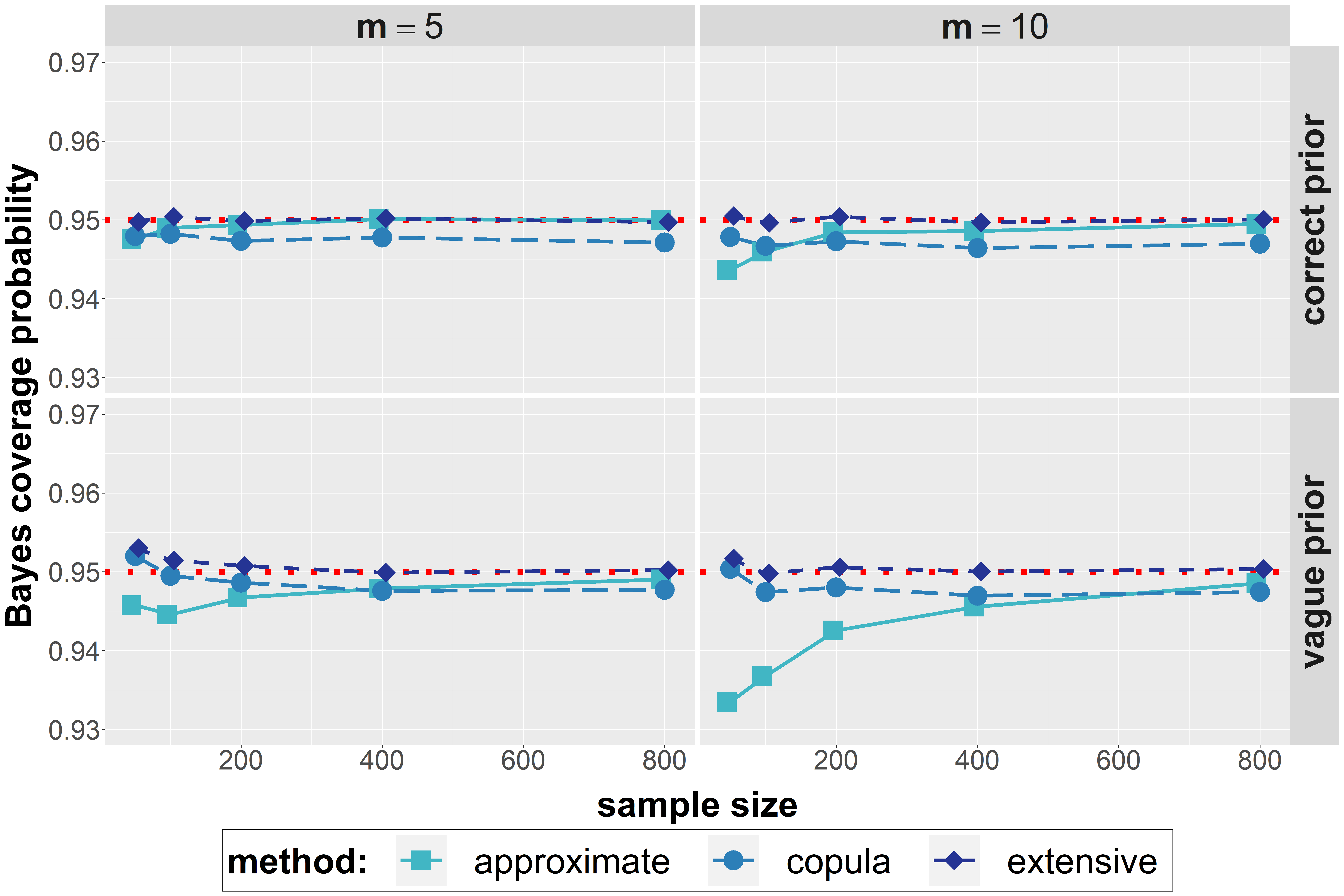}
	}
	\caption{Simulated Bayes coverage probability of different credible regions for raw proportions $\bm \vartheta$. Results are averaged over four different generative distributions for each $m=5, 10$.}
	\label{figure2}
\end{figure*}

Figure \ref{figure2} shows the Bayes coverage probability of the different credible regions for the raw proportions $\bm \vartheta$. The results are only stratified for the number of proportions $m$ and whether the correct or a vague prior is used for the analysis. In effect, each simulated BCP in figure \ref{figure2} is the average over all four scenarios ($\nu \in\{ 20, 40\}$, $\rho \in \{0.5, 0.75\}$) and is thus comprised of $200,000=4\cdot 50,000$ simulation runs.

If the analysis prior corresponds to the true generative prior distribution, all methods have close to target coverage level for $m=5$. As the dimension increases to $m=10$, the BCP deviates more from the target level $95\%$. If the vague analysis prior is employed, the normal approximation clearly performs worse compared to the copula and the extensive approach, in particular for low sample sizes.
More detailed results in appendix \ref{details} suggest that the normal approximation becomes worse not only as $m$ increases but also as the concentration $\nu$ or the correlation $\rho$ decrease. This is plausible as in both cases parameter values $\vartheta_j$ near the boundaries of the unit interval become more likely which negatively inflects the quality of the normal approximation.

\begin{figure*}%[ht!]
	\centerline{%
		\includegraphics[width=.9\linewidth]{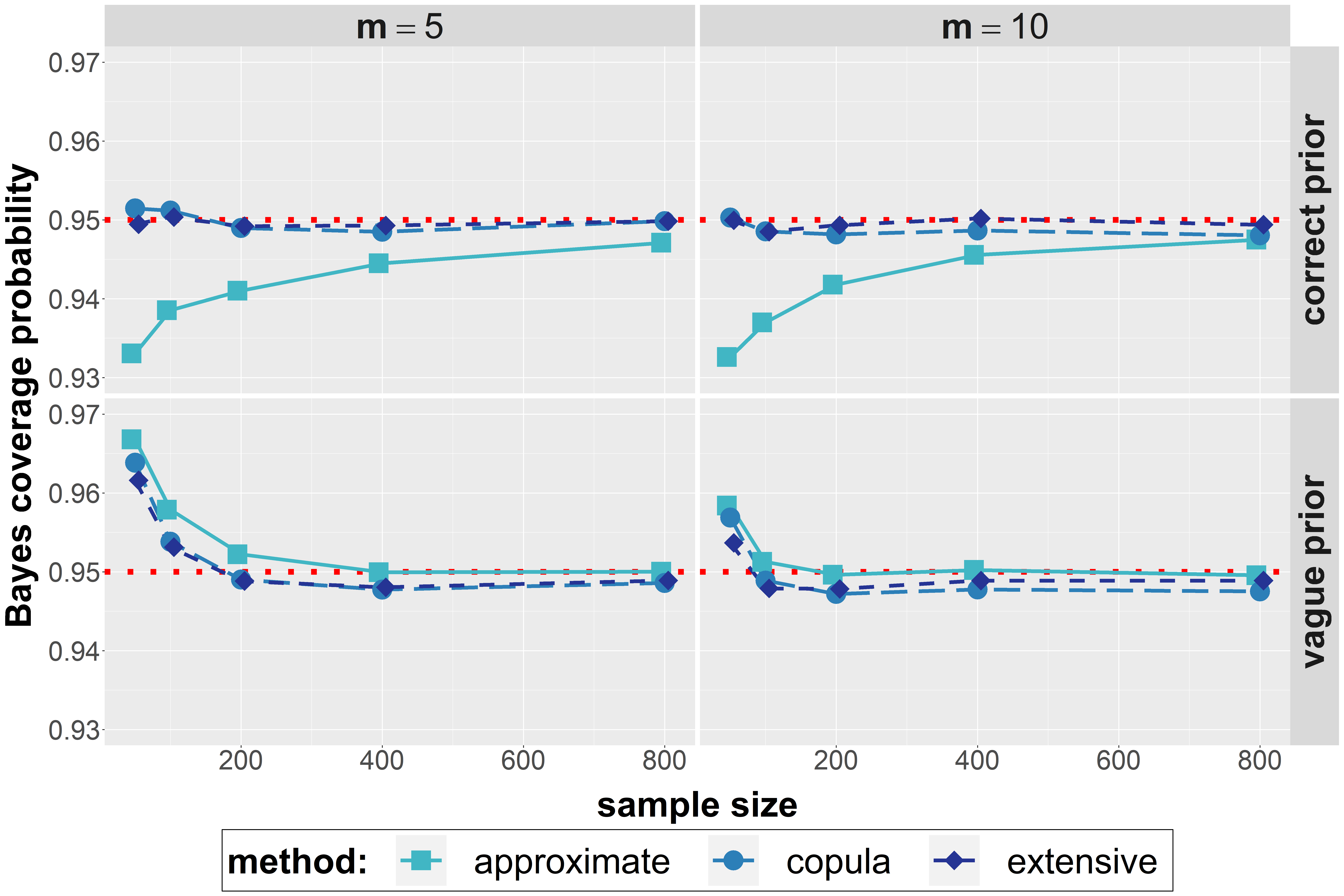}
	}
	\caption{Simulated Bayes coverage probability of different credible regions for differences of proportions $\vartheta_j-\vartheta_m$. Results are averaged over four different generative distributions for each $m=5, 10$.}
	\label{figure3}
\end{figure*}

Figure \ref{figure3} shows similar investigations for credible regions for differences of parameters $\vartheta_j-\vartheta_m$, ${j=1,\ldots,m-1}$. Overall, the picture is similar to the previous analysis. The deviations from the target coverage probability are larger for small $n$ (compared to large $n$) and when the normal approximation is used (compared to the other two approaches). The main difference is that the actual BCP is larger than the target $1-\alpha$ for small $n$ when the vague prior is employed. We attribute this observation to the fact that the (induced) prior for the difference $\vartheta_j-\vartheta_{m}$ is no longer a uniform but rather a triangular distribution.

Besides the coverage probabilities we also investigated the frequency to obtain a credible region ${\CR \not \subset (0,1)^m}$ not entirely in the support of the distribution. While for the copula and extensive approach this probability is zero for all sample sizes by construction, for the approximate method it is nonzero. For the raw proportion analysis (figure 2) this probability is as low as $42\%$ (vague prior) or $78\%$ (correct prior) for $n=50$ and $m=10$ and does stabilize to at least $94\%$ for all scenarios where $n\geq 200$.

%%%%%%%%%%%%%%%%%%%%%%%%%%%%%%%%%%%%%%%%%%%%%%%%%%%%%%%%%%%%%%%%%%%%%%%%%%%%%%%%%%%%%%%%%%%%%%%%%%%%%
\section{Discussion}\label{discussion}

\subsection{Summary}

In this work, a simple construction of a bivariate Beta distribution from a four-dimensional Dirichlet distribution due to \citet{olkin2015} was generalized to higher dimensions. As $2^m$ parameters are needed to describe the $m$-dimensional Beta distribution, it is of limited to no use in high dimensions. To counter this problem, a reduced parametrisation only requiring $1+m(m+1)/2$ parameters was proposed which can be derived from an overall concentration parameter $\nu$ a mean vector $\bm \mu$ and a correlation matrix $\bm R$. Necessary and sufficient conditions have been provided that need to be satisfied such that an mBeta distribution for a given triple $(\nu, \bm \mu, \bm R)$ can exist. 

These moment conditions \eqref{mc} provide some intuition on which correlation structures are admissible. However, they are also of limited use in practice as checking the conditions for general $\bm \mu$ and $\bm R$ is usually not feasible, at least not in a numerically efficient manner. A more concrete descriptions of these conditions may be obtained by calculating the extreme rays of the polyhedral cone $\{\bm b: \widetilde{\bm H}\tra \bm b \geq \bm 0 \}$. That is, one would need to compute the so-called V-representation $\{\bm B \bm \lambda\, |\, \bm \lambda \geq \bm 0 \}$ which implies a finite but potentially large number of conditions $\bm b \tra \widetilde{\bm \alpha} \geq 0$, imposed by the columns $\bm b$ of the matrix $\bm B$. In \texttt{R}, algorithms for that matter are for instance implemented in the package \texttt{rcdd} \citep{rcdd}. 
As the number of generating rays is rapidly growing in the dimension $m$, this method is again only helpful for small dimensions. A similar approach was recently pursued by \citet{fontana2018} in a related context.
Altogether, it appears that the verification of the validity of a multivariate Beta distribution in terms of its mean $\bm \mu$ and correlation structure $\bm R$ is only feasible in higher dimensions when making simplifying structural assumptions, see section \ref{model}. This is only a concern for the prior distribution, as we are guaranteed to end up with a valid posterior when we start with a valid prior (proposition \ref{practice}).

An example of a valid prior was the vague prior employed in the simulation study which corresponds to independent uniform prior distributions. This case can be connected to the so-called Bayes prior which is frequently employed for the Bayesian analysis of a single proportion \citep[p.\,173]{ASI}. Marginally the two approaches do the same, namely adding two pseudo-observations (one success, one failure) to the dataset leading to shift of posterior mass towards $1/2$. The proposed mBeta model additionally includes a prior on and update of the mixed second-order moments. From a frequentist viewpoint, this approach can thus be used for a joint shrinkage estimation of sample mean and covariance matrix as they both depend only on the posterior parameters $(\nu^*, \bm A^*)$ (proposition \ref{properties}). 

The idea to connect marginal Beta-binomial models via a copula approach is not new. In previous work following this direction, marginal distribution and copula parameters were usually jointly estimated \citep{nyaga2017, yamaguchi2019}. In contrast, in this work the copula model is only fitted to the joint multivariate Beta posterior distribution to construct simultaneous credible regions.

The simulation study indicates that the copula and extensive sampling approaches result in credible regions with close to the desired Bayes coverage probability. 
% V4: 
A disadvantage of our proposal which relies on a Gaussian copula model is that the resulting posterior correlation matrix is not equal to the actual correlation matrix according to the mBeta-binomial model. 
This is due to the fact that the nonlinear copula transformation does not preserve the linear correlation between variables.
This issue is more severe for small posterior concentration values $\nu^*=\nu+n$ and vanishes asymptotically as $n \rightarrow \infty$ (Bernstein–von Mises theorem).
We are not aware of a simple correction strategy to address this problem other than to rely on simulations which is numerically expensive.  
% V4 end
In the situations that were assessed in the simulation study ($n\geq 50$), the loss in accuracy when employing the copula approach compared to the extensive sampling approach seems to be negligible. 
In contrast, the normal approximation requires a much larger sample size for satisfactory results and additionally does not guarantee credible regions which lie entirely in the parameter space. As it provides no benefits compared to the copula approach besides simplicity, the latter seems to be a reasonable default choice for the considered problem.

A methodological limitation of this work is that the adequateness of the reduced relative to the full parametrisation cannot be assessed via numerical simulation in higher dimensions. This is due to the fact that data generation from the full underlying Dirichlet distribution becomes also numerically infeasible for $m$ much larger than 10.

\subsection{Extensions}

The present work focused on the construction of multivariate credible regions. Decisions regarding prespecified hypotheses based the posterior distribution of the parameters may of course also be of interest in practice. \citet{madruga2001} and \citet{thulin2014} connect credible regions to hypothesis testing in a decision theoretic framework. They show that for specific loss functions, the standard Bayes test, i.e. deciding for the hypothesis with lowest posterior expected loss, corresponds to comparison of parameter values with credible bounds which may be denoted as $\bm \varphi_{CR}=(\one(\vartheta_0 \notin \CR_{1-\alpha}^{(j)}))_{j=1,\ldots,m} \in \{0,1\}^m$.

An advantage of the approach of \citet{thulin2014} is that the employed loss function does not depend on the observed data. This was a non-standard feature of the proposal by \citet{madruga2001}. It appears that the approach of \citet{thulin2014} can however not easily be transferred to the multivariate setting which would require to specify a loss function $L:  \bm \Theta \times \bm \setA  \rightarrow \R$ such that
\begin{align}
\bm \varphi_{CR} = \argmin_{\bm \varphi \in \bm \setA} \E_{\pi^*}L(\bm \vartheta, \bm \varphi).
\end{align}
Hereby, the action space $\bm \setA$ consists of all possible test decisions, e.g. $\bm \setA=\{0,1\}^m$ for the one-sided hypothesis system
\begin{align}
\setH = \{H_j:\ \vartheta_j \leq \vartheta_0,\ j=1,\ldots,m\},
\end{align} 
A possible generalization of the loss function $L_{(2)}$ from \citet[p.\,136]{thulin2014} is
\begin{align}\label{loss_candidate}
L(\bm \vartheta, \bm \varphi) = || \bm \varphi \odot (\bm 1_m -\bm \chi) ||_{\infty} (1-\alpha) + || (  \bm 1_m- \bm \varphi) \odot\bm \chi ||_{\infty}  \alpha,
\end{align}
whereby $\bm \chi =\left(\one(\vartheta_j \in K_j)\right)_{j=1,\ldots,m} $ indicates for all $j=1,\ldots,m$ whether $\vartheta_j$ is contained in the alternative $K_j = \Theta \setminus H_j$. \citet{thulin2014} showed that $\bm \varphi_{CR}$ is a Bayes test under the loss function \eqref{loss_candidate} in the case $m=1$. In the multivariate setting ($m>1$), this no longer holds true which can be confirmed via numerical examples.
However, $\bm \varphi_{CR}$ can be seen as a constrained Bayes test. That is to say, from all tests with posterior false positive probability $\pr_{\pi^*}(|| \bm \varphi \odot (\bm 1_m -\bm \chi) ||_{\infty} = 1)$ bounded by $\alpha \in (0,1)$, it minimizes the posterior false negative probability $\pr_{\pi^*}(|| (  \bm 1_m- \bm \varphi) \odot\bm \chi ||_{\infty} = 1)$. 

These considerations are somewhat opposing the usual (empirical-) Bayes approach to multiplicity adjustment which is usually based on modifying the prior distribution rather than the loss function \citep{scott2009, guo2010, scott2010}. This established strategy could also be employed for the multinomial Beta-binomial model considered in this work. For instance the prior distribution could be modified such that the tail probability {$u(m) = \pr_{\pi}(||\bm \vartheta||_\infty > \vartheta_0)$} is controlled, e.g. by increasing the concentration parameter $\nu$ (assuming $\bm \mu < \vartheta_0$). Under the vague (independent uniform) prior employed in chapter \ref{sim}, such a control is not given as $u(m) = 1-\vartheta_0^m \rightarrow 1$ for $m\rightarrow \infty$.
The above sketched approaches (adaptation of the loss function, e.g. \eqref{loss_candidate}) to multiple hypothesis testing in the Bayesian framework should be contrasted thoroughly with the established methods (adaptation of the prior, hierarchical models) in the future.

% ** Acknowledgements **
\section*{Acknowledgements}
This project was funded by the Deutsche Forschungsgemeinschaft (DFG, German Research Foundation) - Project number 281474342/GRK2224/1.

\section*{Conflict of interest}
The author declares that there is no conflict of interest.

\bibliography{literature}  %%% Remove comment to use the external .bib file (using bibtex).

\newpage
\appendix

\setcounter{section}{0}

%%%%%%%%%%%%%%%%%%%%%%%%%%%%%%%%%%%%%%%
%\vspace{0.0cm}
\section{\textbf{Technical details}}
\label{proofs}
\vspace{.25cm}

At first, some established results concerning the Dirichlet distribution are stated. Let $\bm p = (p_1,\ldots,p_w)\tra \sim \Dir(\bm \gamma)$ with support $\bm \setP= \{\bm p \in (0,1)^w: ||\bm p||_1=1\}$ whereby $\bm \gamma = (\gamma_1,\ldots,\gamma_w)\tra  \in \mathbb{R}^w_+$ and $\nu=|| \bm \gamma||_1$. %and $w=2^m, m \in \mathbb{N}$.
An essential property of the Dirichlet distribution is the so-called aggregation property \citep[Theorem 2.5 (i)]{DIR}. It concerns the vector $\tilde{\bm  p}$ where two components $p_k$ and $p_{k'}$ from the original $\bm p$ are replaced by their sum which has the following distribution:
\begin{align}\label{aggregation}\tag{A.1}
(p_1,\ldots,p_k + p_{k'}, \ldots, p_w) \sim \Dir(\gamma_1,\ldots,\gamma_k + \gamma_{k'}, \ldots, \gamma_w).
\end{align}
Repeated application of this result allows to aggregation of arbitrary subvectors of $\bm p$. In particular, the marginal distribution of the component $p_k$ turns out to be 
\begin{align}
p_k \sim \Beta(\gamma_k,\ \nu-\gamma_k).
\end{align}
Additionally, we will use the fact that $\var(p_k) = \frac{\gamma_k(\nu-\gamma_k)}{\nu^2(\nu+1)}$ and $\cov(p_k, p_{k'})= \frac{-\gamma_k \gamma_{k'}}{\nu^2(\nu+1)}$ for $k \neq k'$ \citep[p.\, 39]{DIR}. When setting $\bm \Gamma = \diag(\bm \gamma)$, this amounts to
\begin{align}\label{dircov}\tag{A.2}
\cov(\bm p) =   \frac{\nu \bm \Gamma - \bm \gamma \bm \gamma\tra}{\nu^2(\nu+1)}.
\end{align}

\begin{proof}[Proof of proposition \ref{properties}]\quad
	\begin{enumerate}
		\item Follows from the definition $\bm X = \bm H \bm C$.
		\item We have $\bm H \in \{0,1\}^{m\times w}$ and $||\bm p||_1=1$ for all $\bm p$. Thus, $\bm \vartheta = \bm H \bm p \in (0,1)^m$.
		\item Is a consequence of the aggregation property \eqref{aggregation}.
		\item Follows from $\vartheta_j \sim \Beta(\alpha_j, \beta_j) \Rightarrow \E(\vartheta_j) = \alpha_j/(\alpha_j+\beta_j)$ 
		and (3).
		\item Follows from \eqref{dircov} and the the fact $\bm \vartheta= \bm H \bm p$ is a linear transformation of $\bm p$.
		\item A generalization of (3) and again a consequence of \eqref{aggregation}. The only thing left to check is that the correct parameters $\gamma_k$ are added up.
	\end{enumerate}
\end{proof}

For the proof of proposition \ref{main}, Farkas' lemma will be employed which is stated below \citep[p.\,263]{CO}. As usual, inequalities between vectors should be interpreted component-wise.

\begin{lemma}[Farkas' lemma]\label{farkas}
	For any matrix $\bm A \in \R^{n \times m}$ and vector $\bm b \in \R^m$, the following two statements are equivalent:
	\begin{enumerate}
		\item The linear system of equations $\bm A \bm x = \bm b$ is feasible, i.e. has a solution $\bm x \in \R^m$, with $\bm x \geq \bm 0$.
		\item For all $\bm y \in R^n$, $\bm A \tra \bm y \geq \bm 0$ implies $\bm y \tra \bm b \geq 0$. 
	\end{enumerate}
\end{lemma}

\quad

\begin{proof}[Proof of proposition \ref{main}]\quad
	\begin{enumerate}
		\item For a given $\bm A$ or derived moment matrix $\bm A(\nu,\bm \mu, \bm R)$, the linear system \eqref{lp} needs to be solved
		for $\bm \gamma \in \R_+^w, w = 2^m$. Hereby $\widetilde{\bm H}$ and $\widetilde{\bm \alpha}$ are given as in definition \ref{mc_def}. Farkas' lemma implies that \eqref{lp} is feasible if and only if \eqref{mc} holds. 
		\item The linear system \eqref{lp} consists of $r=1+m(m+1)/2$ equations and $w=2^m$ unknowns. Thus, it has a unique solution only for $m=2$ as then $w=r$. On the other hand $m>2$ implies $w > r$ and thus \eqref{lp} is underdetermined and has no unique solution in this case. To enforce uniqueness, the minimization of (e.g.) $||\bm \gamma - \bm 1_w\nu/w||_2$ under the side condition \eqref{lp} can be reformulated as a convex linearly constrained quadratic program and thus has has a unique solution.
		
	\end{enumerate}
\end{proof}

\begin{proof}[Proof of proposition \ref{practice}]\quad
	\begin{enumerate}
		\item It was shown in the proof of proposition \ref{main} that \eqref{mc} is equivalent to the feasibility of \eqref{lp}. Hence, specification of a feasible solution $\bm \gamma$ of \eqref{lp} implies \eqref{mc}.
		\item Follows immediately from (1), as $\bm d \in \R_+ ^w$.
		\item Let $\bm b \in \R^{1+m(m+1)/2}$ with $\widetilde{\bm H}\tra \bm b \geq \bm 0$ be given. 
		Let $\widetilde{\bm u}$ and $\widetilde{\bm \alpha}^*$ be the derived moment vectors from $\bm U$ and $\bm A^* = \bm A + \bm U$, respectively, similar as in definition \ref{mc_def}. Then 
		\begin{align}
		\bm b\tra \widetilde{\bm \alpha}^* = \bm b\tra(\widetilde{\bm \alpha} + \widetilde{\bm u}) =\bm b\tra \widetilde{\bm \alpha} + \bm b\tra \widetilde{\bm u} \geq 0 + 0 = 0.
		\end{align}
		\item For $m=2$, the parameters $\nu = 4$ and 		$\bm A = \begin{pmatrix} 
		2 & 3\\
		3 & 3 
		\end{pmatrix}$ together with the vector ${\bm b = (1,0,-1,0)\tra}$ provide a counterexample as $\bm b\tra \widetilde{\bm \alpha}=-1<0$ but
		\begin{align}
		\widetilde{\bm H}\tra \bm b = \begin{pmatrix} 
		0 & 0 & 0 & 1\\
		0 & 1 & 0 & 1\\
		1 & 0 & 0 & 1\\
		1 & 1 & 1 & 1
		\end{pmatrix}
		\begin{pmatrix} 
		1 \\ 
		0 \\
		-1\\
		0
		\end{pmatrix}
		= 
		\begin{pmatrix} 
		0 \\ 
		0 \\
		1\\
		0
		\end{pmatrix}
		\geq \bm 0.
		\end{align}
	\end{enumerate}
	
\end{proof}

\begin{proof}[Proof of proposition \ref{mBeta_model}]\quad
	\begin{enumerate}
		\item The first assertion follows from the one-to-one connection between mBeta and Dirichlet distribution and their respective samples ($\mathfrak{X}, \mathfrak{C}$) and the established update rule $\bm \gamma^* = \bm \gamma + \bm d$ for the Dirichlet-multinomial model.
		\item The second result follows by noting that
		\begin{align}
		\bm A^* =& \bm H \bm \Gamma^* \bm H\tra = \bm H \diag(\bm \gamma + \bm d) \bm H\tra \\
		=&\bm H \diag(\bm \gamma) \bm H\tra + \bm H \diag(\bm d) \bm H\tra 	= \bm A + \bm U.
		\end{align}
	\end{enumerate}
\end{proof}

%%%%%%%%%%%%%%%%%%%%%%%%%%%%%%%%%%%%%%%
\newpage
\section{\textbf{Numerical example}}\label{example}
\vspace{.25cm}

\begin{example}[mBeta update rule]\label{update}
	Three proportions $\bm \vartheta = (\vartheta_1,\vartheta_2,\vartheta_3 )\tra$ shall be assessed on the same dataset. Our prior belief in terms of the mean vector is $\bm \mu = (0.8, 0.775, 0.75)\tra$. The parameters are assumed to be positively correlated, modeled as an equicorrelation of ${\rho=0.5}$. Our certainty in this prior is limited, expressed by an overall concentration parameter $\nu = 20$ which can be interpreted as the prior sample size. This can be modeled as $\bm \vartheta \sim \mBeta(\bm \gamma)$ with
	\begin{align}
	\bm \gamma &= (2.57,\, 0.00,\,  0.16,\,  1.27,\,  0.36,\,  1.57,\,  1.91,\, 12.17)\tra\\
	\quad \Rightarrow \quad \nu &= 20,\quad \bm A = \begin{pmatrix}
	16.00 &14.07 &13.73\\
	14.07 &15.50 &13.43\\
	13.73 &13.43 &15.00
	\end{pmatrix},
	\end{align}
	whereby $\bm A$ represents the reduced parametrisation. The observed experimental data
	\begin{align}
	\bm d &= (24,\, 10,\, 0,\, 29,\, 9,\, 8,\, 58,\, 179)\tra\\
	\quad \Rightarrow \quad n &=317, \quad \bm U = \begin{pmatrix}
	254  &237  &187 \\
	237  &266  &208 \\
	187  &208  &226
	\end{pmatrix},
	\end{align}
	leads to a posterior distribution $\bm \vartheta \given \bm d \sim \mBeta(\bm \gamma^*)$ with
	\begin{align}
	\bm \gamma^* &= \bm \gamma + \bm d \\
	&= (26.57,\,10.00,\,0.16,\,30.27,\,   9.36,\,9.57,\,59.91,\,191.17)\tra
	\end{align}
	and reduced parametrisation
	\begin{align}
	\nu^* &= \nu + n = 337,\\
	\bm A^* &= \bm A +\bm U =  \begin{pmatrix}
	270.00 &251.07 &200.73 \\
	251.07 &281.50 &221.43 \\
	200.73 &221.43 &241.00
	\end{pmatrix}.
	\end{align}	
	This can be translated back to posterior mean and correlation matrix 
	\begin{align}
	\bm \mu^* = (0.80,\, 0.84 ,\,0.72)\tra, \quad \bm R^* = \begin{pmatrix}
	1.00 &0.51 &0.13\\
	0.51 &1.00 &0.36\\
	0.13 &0.36 &1.00
	\end{pmatrix}.
	\end{align}
	A visualization of prior and posterior distribution of this example is provided in figure \ref{figure1} in section \ref{model}.	
\end{example}

%%%%%%%%%%%%%%%%%%%%%%%%%%%%%%%%%%%%%%%
\newpage

\renewcommand\thefigure{\thesection.\arabic{figure}}    
\setcounter{figure}{0}  

\section{\textbf{Additional simulation results}}\label{details}
%\vspace{.25cm}

\subsection{\textbf{Analysis of raw proportions}}

\begin{figure*}[h!]
	\centerline{%
		\includegraphics[width=0.85\linewidth]{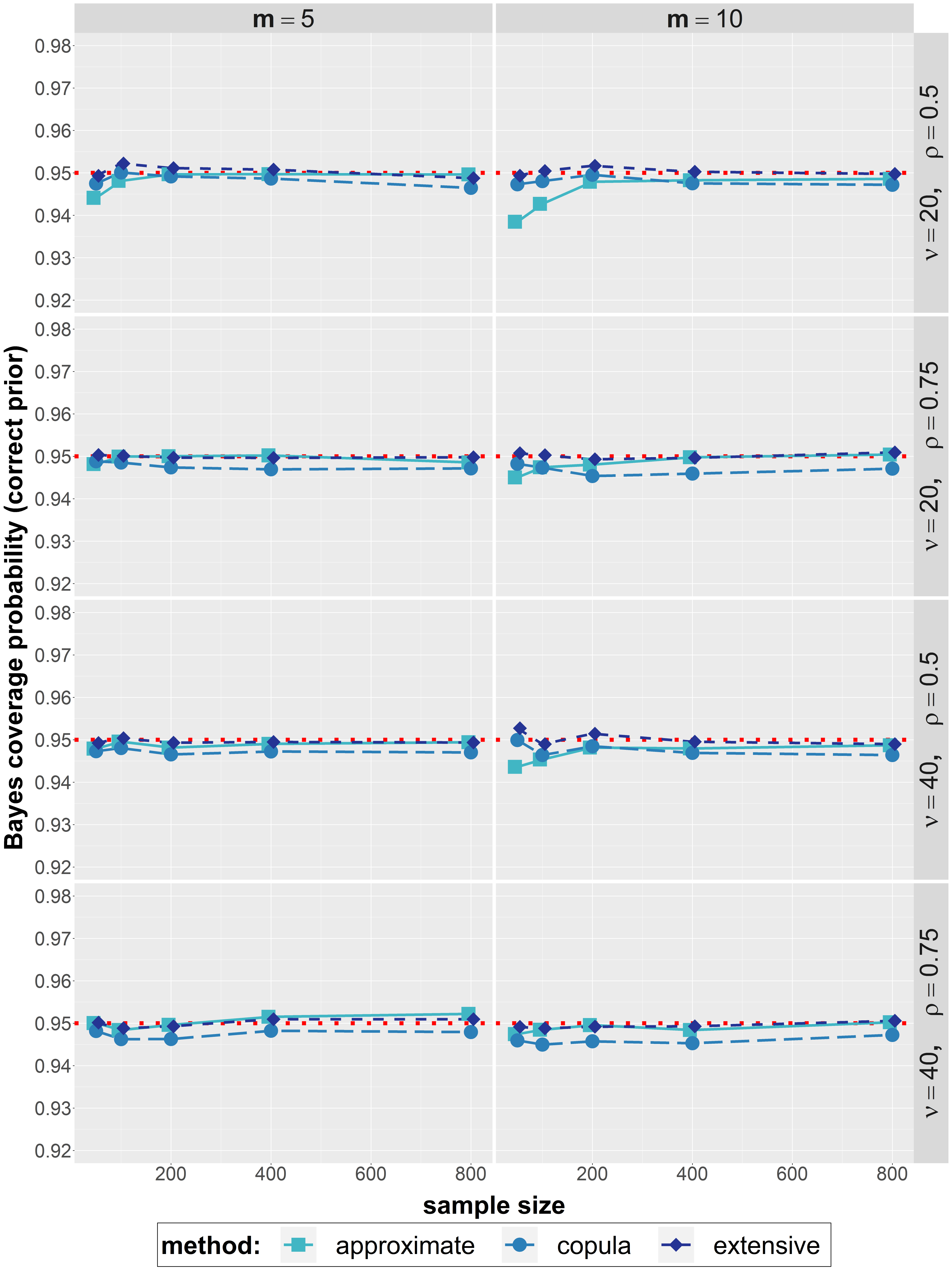}
	}
	\caption{Simulated Bayes coverage probability of different credible regions for raw proportions $\vartheta_j,\ j=1,\ldots,m$. Results are stratified by generative prior distribution, see section \ref{setup}. Only results for the correct analysis prior are shown. Each point is based on $50,000$ simulations. 
	}
	\label{figureA2c}
\end{figure*}

\begin{figure*}[h!]
	\centerline{%
		\includegraphics[width=0.85\linewidth]{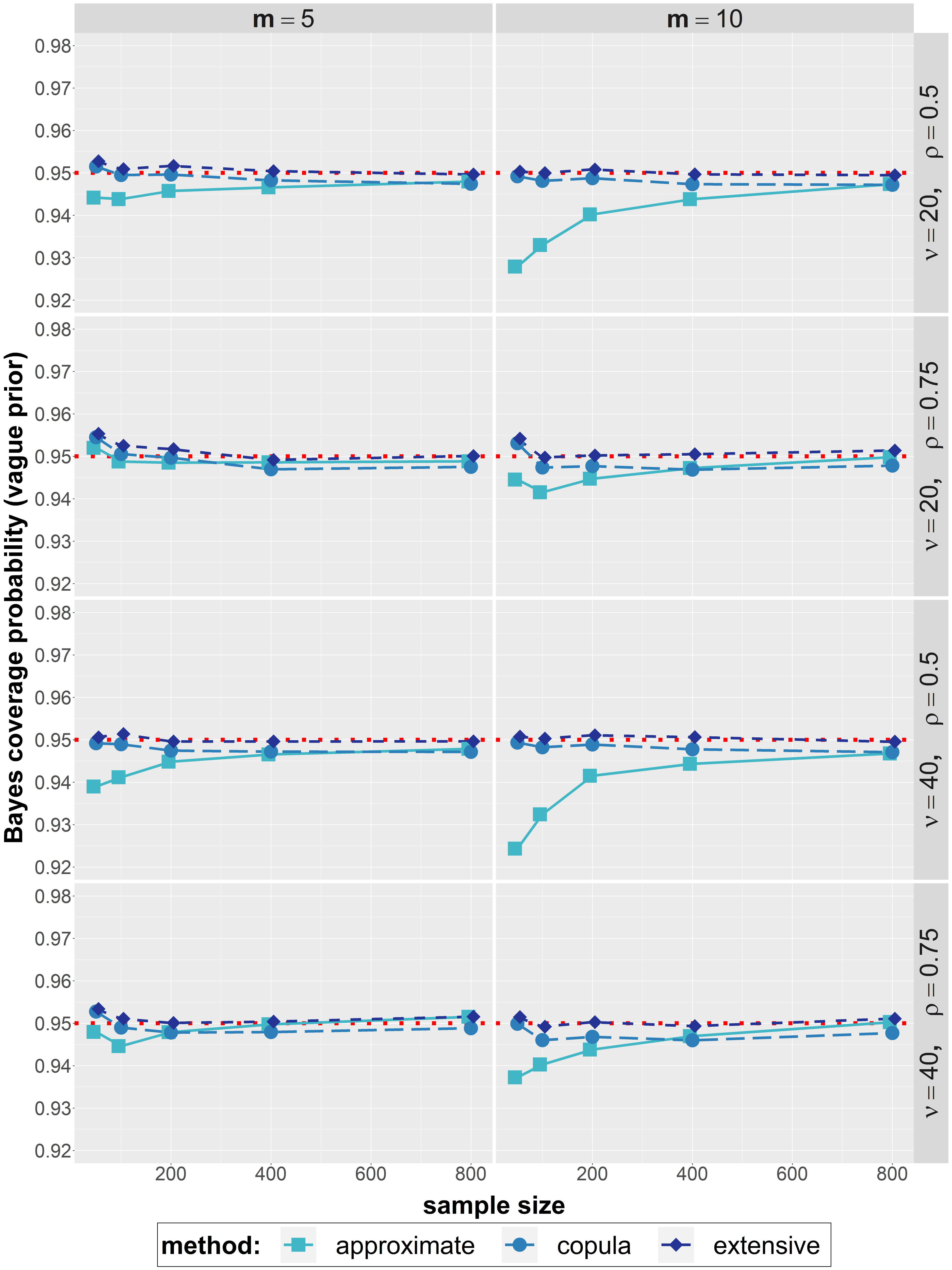}
	}
	\caption{Simulated Bayes coverage probability of different credible regions for raw proportions $\vartheta_j,\ j=1,\ldots,m$. Results are stratified by generative prior distribution, see section \ref{setup}. Only results for the vague analysis prior are shown. Each point is based on $50,000$ simulations.  
	}
	\label{figureA2v}
\end{figure*}

\clearpage

\subsection{\textbf{Analysis of differences of proportions}}

\begin{figure*}[h!]
	\centerline{%
		\includegraphics[width=0.85\linewidth]{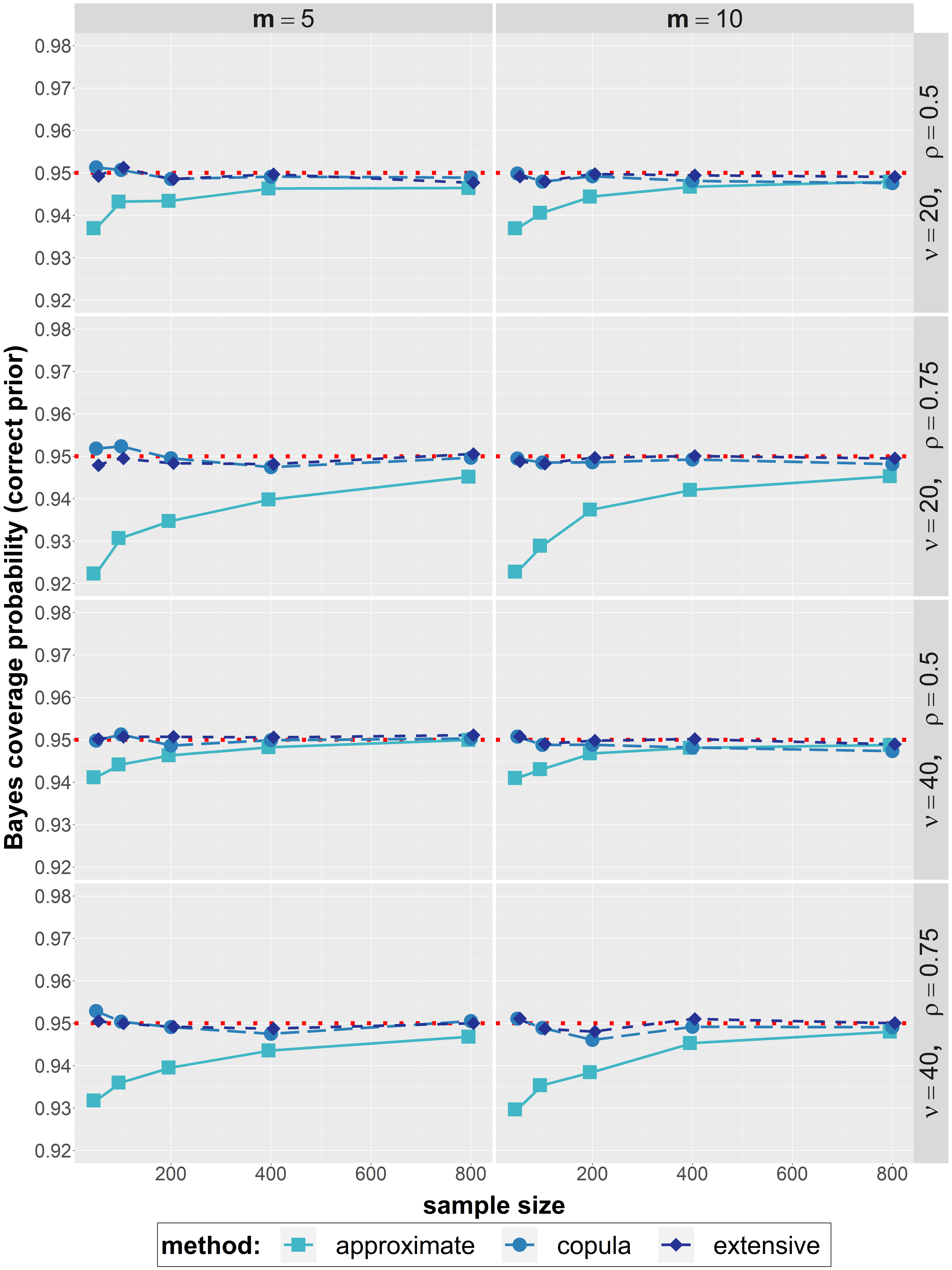}
	}
	\caption{Simulated Bayes coverage probability of different credible regions for differences of proportions $\vartheta_j - \vartheta_m$, $j=1,\ldots,m-1$. Results are stratified by generative prior distribution, see section \ref{setup}. Only results for the correct analysis prior are shown. Each point is based on $50,000$ simulations. 
	}
	\label{figureA3c}
\end{figure*}

\begin{figure*}[h!]
	\centerline{%
		\includegraphics[width=0.85\linewidth]{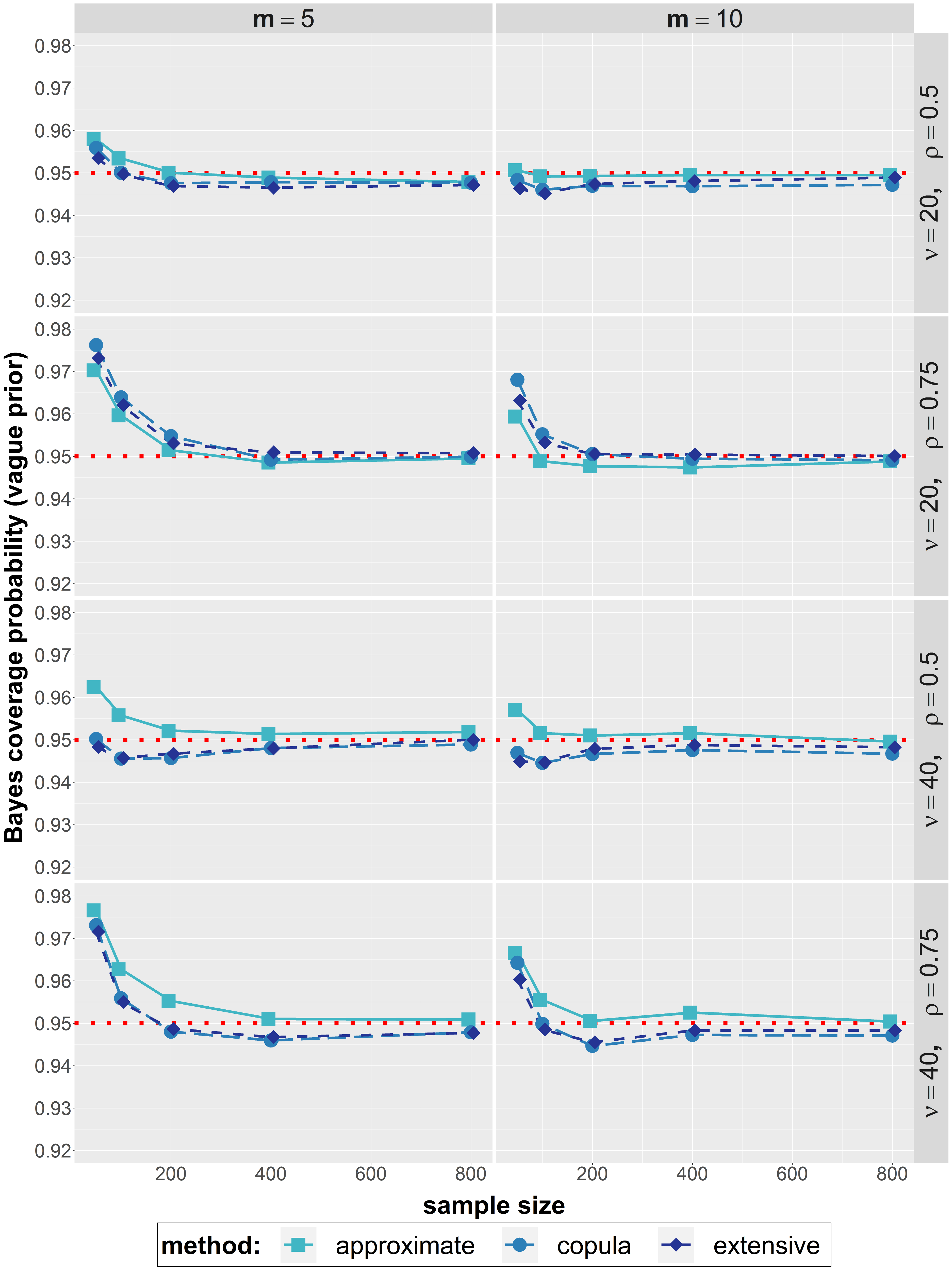}
	}
	\caption{Simulated Bayes coverage probability of different credible regions for differences of proportions $\vartheta_j - \vartheta_m$, $j=1,\ldots,m-1$. Results are stratified by generative prior distribution, see section \ref{setup}. Only results for the vague analysis prior are shown. Each point is based on $50,000$ simulations. 
	}
	\label{figureA3v}
\end{figure*}

\clearpage

\end{document}